\documentclass[conference]{IEEEtran}

\IEEEoverridecommandlockouts

\makeatother
\usepackage{graphicx}
\usepackage{tikz}
\usepackage{color}
\usepackage{amsmath, amsfonts, amssymb, mathrsfs}
\usepackage[amsmath,thmmarks]{ntheorem}
\usepackage{subfigure}
\usepackage{epstopdf}
\usepackage{pgfplots}
\usepackage{arydshln}
\usepackage{mathabx}
\usepackage{cite}
\usepackage[super]{nth}
\usepackage{cite}
\usepackage{todonotes}
\usepackage{subfigure}
\usepackage{algorithm}
\usepackage{algorithmic}
\usepackage{bbm}
\makeatother

\theoremstyle{plain}
\theoremheaderfont{\itshape\bfseries}
\theorembodyfont{\normalfont\itshape}
\theoremseparator{:}
\theoremsymbol{}
\newtheorem{theorem}{Theorem}
\newtheorem{lemma}[theorem]{Lemma}
\newtheorem{corollary}[theorem]{Corollary}
\newtheorem{proposition}[theorem]{Proposition}

\theoremstyle{nonumberplain}

\theoremstyle{plain}
\theoremheaderfont{\itshape\bfseries}
\theorembodyfont{\normalfont}
\theoremseparator{:}
\theoremsymbol{}
\newtheorem{remark}{Remark}

\theoremstyle{plain}
\theoremheaderfont{\itshape\bfseries}
\theorembodyfont{\normalfont}
\theoremseparator{:}
\newtheorem{example}{Example}
\newtheorem{assum}{Assumption}

\theoremstyle{nonumberplain}
\theoremsymbol{\rule{1ex}{1ex}}

\newtheorem{proof}{Proof}
\newlength\fheight
\newlength\fwidth
\newcommand{\abs}[1]{\left| #1 \right|}
\newcommand{\norm}[1]{\| #1 \|}
\newcommand{\inn}[2]{\langle #1,#2 \rangle}
\newcommand{\lrbrace}[1]{\left\{ #1 \right\}}

\newcommand{\normi}[1]{{\left\vert\kern-0.25ex\left\vert\kern-0.25ex\left\vert #1 
		\right\vert\kern-0.25ex\right\vert\kern-0.25ex\right\vert}}
\newcommand{\inni}[2]{{\langle\kern-0.25ex\langle #1,#2
		\rangle\kern-0.25ex\rangle}}

\def\X{\mathcal{X}}

\def\d{ \mathrm{d} }								
\def\T{ \mathrm{T} }								

\def\nat{ \mathbb{N} }								

\def\real{ \mathbb{R} }								

\def\Erw{ \mathbb{E} }

\def\Prob{ \mathbb{P} }

\def\X{\mathcal{X}}

\def\M{\mathbf{M}}

\def\ACV{\textnormal{ACV}}
\def\AD{\textnormal{AD}}
\def\Del{\textnormal{D}}

\newcommand\restr[2]{{
  \left.\kern-\nulldelimiterspace 
  #1 
  \right|_{#2} 
  }}

\DeclareMathOperator*{\argmax}{arg\,max}

\makeatletter
\newcommand\fs@betterruled{%
	\def\@fs@cfont{\bfseries}\let\@fs@capt\floatc@ruled
	\def\@fs@pre{\vspace*{5pt}\hrule height.8pt depth0pt \kern2pt}%
	\def\@fs@post{\kern2pt\hrule\relax}%
	\def\@fs@mid{\kern2pt\hrule\kern2pt}%
	\let\@fs@iftopcapt\iftrue}
\floatstyle{betterruled}
\restylefloat{algorithm}
\makeatother
\title{Resource-Aware Control via Dynamic Pricing for Congestion Game with Finite-Time Guarantees}
\author{\IEEEauthorblockN{Ezra Tampubolon$^\dagger$, Haris Ceribasic$^\dagger$, and Holger Boche$^\dagger$*}
	\IEEEauthorblockA{
		$^\dagger$Technische Universit{\"a}t M{\"u}nchen, Lehrstuhl f{\"u}r Theoretische Informationstechnik\\
		*Munich Center for Quantum Science and Technology (MCQST)\\
		\{ezra.tampubolon,haris.ceribasic,boche\}@tum.de}
	
}

\begin{document}
%
\maketitle
%
\begin{abstract}
Congestion game is a widely used model for modern networked applications. A central issue in such applications is that the selfish behavior of the participants may result in resource overloading and negative externalities for the system participants. In this work, we propose a pricing mechanism that guarantees the sub-linear increase of the time-cumulative violation of the resource load constraints.
The feature of our method is that it is resource-centric in the sense that it depends on the congestion state of the resources and not on specific characteristics of the system participants. This feature makes our mechanism scalable, flexible, and privacy-preserving. Moreover, we show by numerical simulations that our pricing mechanism has no significant effect on the agents' welfare in contrast to the improvement of the capacity violation.
\end{abstract}

\section{Introduction}
\label{sec:intro}
Optimizing users/devices competing for utilization of resources (be it network link, power supply, or wireless spectrum) have become essential components of modern networked systems such as IoT, smart grid, and cognitive radio. A trend in recent years is that the number of users in such applications increases tremendously (see e.g. \cite{Evans2011}). For instance: Analysts predicts that more than 50 billion things are expected to be connected over the internet by the end of 2020 \cite{Evans2011}. Such rapid growth involves certainly a series of challenges. 
One of the main challenges facing the system managers is congestion control (CC) of the available resources. For without it, negative externalities in the form of quality degradation of resource services might occur due to overload. For instance, in wireless communication network applications, an excessive amount of traffic through a base station or an access point (resource) might result in buffer bloat, and consequently in the inefficiency of the system in the form of high latency and network throughput reduction, causing a negative experience for all users. Furthermore, sophisticated congestion control method is crucial for making the electrical power driven technologies environment-friendly, Another trend visible in recent years is that power consumption due to technical applications constitutes a non-negligible part of the global power consumption with the tendency of enormous growth (see e.g., \cite{Pickavet2008}). 

The excessive numbers of users in modern networked systems and the decreasing degree of cooperativeness  justify the attractiveness of the famous game theoretical concepts for system modeling. A natural fundament for developing a CC method is the concept of the congestion game introduced in \cite{Rosenthal1973,Schmeidler1973}.
The corresponding model assumes non-cooperative rational participants, whose strategy is an allocation policy over resources, and whose loss depends proportionally on the total load of the utilized resources. 

The most prominent classical example of a congestion game is the traffic routing model \cite{Wardrop1952}, where the arcs in a given network represent the resources, the different origin-destination pairs specify the player, and the possible action of a player is the allocation over the paths in the system.
This concept has also lead to fruitful discussions in the wireless network literature. It has recently been used in wireless network modeling, e.g., access point selection in WiFi networks \cite{Ercetin2008,Chen2010},
uplink resource allocation  in  multichannel  wireless  access networks  \cite{Altman2009},  wireless  channels  with  multipacket  reception capability \cite{Sanyal2010}, and the impact of interference set in studying the congestion  game  in  wireless  mesh  networks \cite{Argento2010}.

Many CC methods are user-centric in the sense that they require observability of system participants' actions and behaviors and provide specific instructions for all of the system users. Such practicas are not suitable for modern large-scale applications. 
The reason is threefold: First, such methods often lack scalability and flexibility; Second, the typically high number of participants in such applications makes the approaches computationally infeasible; Third, due to growing users' demands of sovereignty and privacy in recent years, direct observation and influence of users' acts by higher authority are highly undesirable.

\paragraph*{Our Contributions} 
In this work, we assume that the agents are rational and cost-oriented, in the sense that they choose actions minimizing the accumulated historical costs, and that they are non-cooperative, i.e. they do not mutually communicate. Based on that, we propose resource-centric dynamic pricing that offers the system participants appropriate incentives to adhere to the resource constraints jointly support sustainable use of the resources. 
We present our theoretical guarantee that our proposed method ensures that the average violation of the capacity constraints decays of order $\mathcal{O}(n^{-1/2})$ w.r.t. the time-horizon $n$. Complementary to this result, we provide numerical simulations for the network routing game -- an instance of the congestion game. As a by-product of our practical investigation, we observe that, compared to the gain in resource sustainability,
our pricing mechanism does not significantly effectuate the agents' welfare (expressed by their average loss significantly).
\paragraph*{Relation to prior work}

The congestion game has been investigated in several directions. Closely related to our work are the following approaches which
consider the game played repeatedly:
Under different model of the individual agents, \cite{Sandholm2001,Fischer2004,Blum2006,Krichene2015} study the convergence of selfish behaviour toward the Nash equilibrium. Besides the fact that it yields a sub-optimal welfare of the agents \cite{Roughgarden2002}, Nash equilibrium of this sort of game might not be a resource sustainable population state (see also the notion of generalized Nash equilibrium in \cite{Facchinei2007,Scutari2012}).  
In order to relieve those undesired effects, several works introduce exciting approaches. Closely related to ours are pricing based methods, e.g. \cite{Alpcan2002,Scutari2006,Ozdaglar2007,Farokhi2015,Barrera2015,Paccagnan2017}. 
The common aspect of the listed works is that they design a population dynamic which converges to the corresponding (designed) equilibrium fulfilling the capacity constraints (see e.g., the concept of generalized Nash equilibrium \cite{Facchinei2007}) of the problem-specific potential game \cite{Monderer1996}. A clear contrast to our work is that they only provide asymptotic guarantee. Moreover, the methods proposed in some of those works, such as,  requires agents' personalized information such as their utilities. 


\paragraph*{Basic Notions and Notations}
For a real vector $a$, $[a]_{+}$ denotes the vector whose entries are the non-negative part of the entries of $a$.
Let $(\X,\norm{\cdot})$ be a normed space and $A,B\subseteq \X$. We denote $A-B:=\lrbrace{x-y:~x\in A,~y\in B}$, and $\norm{A}:=\sup_{x\in A}\norm{x}$. In this work we assume that a probability space $(\Omega,\Sigma,\Prob)$ and a filtration $\mathbb{F}:=(\mathcal{F}_{n})_{n\in\nat_{0}}$ therein are given. 


\section{Setting}
\paragraph*{Congestion Game}
A congestion game consists of a finite set of agents/players $[N]$ and a finite set $\mathcal{R}$ of resources. To each agent $i\in [N]$, there corresponds a collection $\mathfrak{P}_{i}\subseteq 2^{\mathcal{R}}$ of resource bundles. One may encode the latter assumption by defining the adjacency matrix $\M^{(i)}\in \real^{|\mathcal{R}|\times |\mathfrak{P}_{i}|}$ whose $\mathcal{P}_{i}$-th column provides the information about all the resources contained in the bundle $\mathcal{P}_{i}$, i.e.:

\small
\begin{equation*}
[\M^{(i)}]_{r,\mathcal{P}_{i}}=
\begin{cases}
1 \quad &r\in\mathcal{P}_{i}\\
0 &\text{else}
\end{cases}.
\end{equation*}
\normalsize

 The aim of each agent $i\in [N]$ is to execute a certain amount $m_{i}>0$ of tasks by utilizing the bundles of resources from $\mathfrak{P}_{i}$. 
We describe the corresponding (utilization) action/strategy of agent $i$ by a vector $x^{(i)}\in \mathcal{X}_{i}$, where $\mathcal{X}_{i}$ is a scaled simplex on $\mathfrak{P}_{i}$, i.e.:
$
\mathcal{X}_{i}:=\lrbrace{x^{(i)}:=(x^{(i)}_{\mathcal{P}_{i}})_{\mathcal{P}_{i}\in\mathfrak{P}_{i}}\in\real^{\abs{\mathfrak{P}_{i}}}:~\sum_{\mathcal{P}_{i}\in\mathfrak{P}_{i}}x^{(i)}_{\mathcal{P}_{i}}=m_{i}}.
$
For any $\mathcal{P}_{i}\in\mathfrak{P}_{i}$, $x^{(i)}_{\mathcal{P}_{i}}$ corresponds to the amount of tasks agent $i$ allocates to the bundle $\mathcal{P}_{i}$. Equivalently, we can describe the task allocation strategy of agent $i$ by means of the simplex
$\Delta_{i}:=\lrbrace{\mu^{(i)}:=(\mu^{(i)}_{\mathcal{P}_{i}})_{\mathcal{P}_{i}\in\mathfrak{P}_{i}}\in\real^{|\mathfrak{P}_{i}|}:~\sum_{\mathcal{P}_{i}\in\mathfrak{P}_{i}}\mu^{(i)}_{\mathcal{P}_{i}}=1}.$
In this paper we describe the allocation strategy of agent $i$ by means of the simplex $\Delta_{i}$ instead of $\mathcal{X}_{i}$. We denote the set of population strategy by $\Delta=\prod_{i=1}^{N}\Delta_{i}$.



Let $\mu^{(i)}\in\Delta_{i}$ be an allocation action of agent $i$. The total load $\phi_{r}^{(i)}(\mu^{(i)})$ of the resource $r\in\mathcal{R}$ caused by the allocation action $\mu^{(i)}\in\Delta_{i}$ of agent $i$ is given by $\phi^{(i)}_{r}(\mu^{(i)})=\sum_{\mathcal{P}_{i}\in\mathfrak{P_{i}}:r\in \mathcal{P}_{i}}m_{i}\mu^{(i)}_{\mathcal{P}_{i}}$.
	 Accordingly, the total load $\phi_{r}(x)$ of resource $r$ caused by the population strategy $\mu\in \Delta$ is given by $\phi_{r}(\mu)=\sum_{i=1}^{N}\phi^{(i)}_{r}(\mu^{(i)})$
We sometimes also use the notation $\phi:=(\phi_{r})_{r\in\mathcal{R}}$. Moreover, we consider in this work the case where the load of resources $r\in\mathcal{R}$ is desirable to not exceed the capacity $L_{r}\in \real_{>0}$, i.e. $\phi_{r}(\mu(k))-L_{r}=:\Gamma_{r}(\mu)\leq 0$.

To each resource $r\in\mathcal{R}$, we associate a function $\ell_{r}:\real_{\geq 0}\rightarrow\real$ which quantifies negative externalities induced on the resource $r$ due to load $\phi_{r}(\mu)$. We refer to $\ell_{r}$ as the loss function of the resource $r$. We assume throughout:
\begin{assum}
For all $r\in\mathcal{R}$, $\ell_{r}:\real_{\geq 0}\rightarrow\real$ is continuous, convex, and non-decreasing.
\end{assum}
\begin{assum}[Slater's Condition]
There exists $\hat{\mu}\in\Delta$ s.t. $\Gamma(\hat{\mu})<0$.
\end{assum}
The loss of a bundle $\mathcal{P}_{i}\in\mathfrak{P}_{i}$ (for agent $i$) is correspondingly given by $\ell_{\mathcal{P}_{i}}^{(i)}(\mu)=\sum_{r\in\mathcal{P}_{i}}\ell_{r}(\phi_{r}(\mu))$.
Throughout this work we use the notations $\ell^{(i)}:=(\ell^{(i)}_{\mathcal{P}_{i}})_{\mathcal{P}_{i}\in\mathfrak{P}_{i}}$ and $\ell:=(\ell^{(i)})_{i\in[N]}$.

An example of congestion game is the following:
\begin{example}[Network Routing Game] 
\label{Ex:Network Routing Game}
Given a directed Graph $\mathcal{G}=(\mathcal{V},\mathcal{E})$ with a vertex set $\mathcal{V}$ and edge set $\mathcal{E}\subseteq \mathcal{V}\times\mathcal{V}$. In a routing game, the task of agent $i\in [N]$ is to transport a certain amount of commodity $m_{i}>0$ from a starting point $s^{(i)}\in\mathcal{V}$ to a destination $t^{(i)}\in\mathcal{V}$. To fulfill this task, agent $i$ can use a prescribed collection $\mathfrak{P}_{i}\subseteq 2^{\mathcal{E}}$ of edges that connects $s^{(i)}$ and $t^{(i)}$. To every edge (resource) $e\in \mathcal{E}$ there corresponds a function $c_{e}$ (loss) that maps the total amount flow caused by the transport of commodities on the edge $e$ to a non-negative number determining the delay on $e$, and also a constant $\ell_{e}>0$ which prescribed the amount of flow admissible on edge $e$.
\end{example}
\begin{remark}
	The congestion game which we investigate in this chapter is an instance of the so-called potential games which are games admitting a potential function: a  real-valued function whose unilateral change describes the change in the player's payoffs. Finite player potential game was firstly studied by Rosenthal \cite{Rosenthal1973}, who also recognized its relation to congestion game and systematically investigated by Monderer and Shapley \cite{Monderer1996}. The work \cite{Monderer1996} also provides a generalization of the finite player setting to infinite player setting, which is the subject of our investigations.
\end{remark}
\begin{remark} 
	The infinite player setting can be either seen as a mixed strategy version of finite player setting, or as an approximation of the finite player setting with large number of populations, and is more convenient; since the latter case can be cumbersome to analyze. However, there is a key difference between finite and infinite cases: while the Nash equilibria in finite case are exactly the optimizers of the corresponding potential function, not all equilibria are optimizers \cite{Sandholm2001} in the infinite case; although in that case, all optimizers of the potential function are equilibria.
\end{remark}
\paragraph*{Performance Measures}
Let be $k\in\nat$ and $\mu(\tau)$, $\tau\in [k]_{0}$, be a given sequence of population actions from initial time until time slot $k$.    
To evaluate the population performance in the congestion game we use the following criteria:

We measure the resource sustainability of the population sequential actions $(\mu(\tau))_{\tau\in[k]_{0}}$ by the (norm) of the aggregated admissible flow violation defined by:

\small
\[
\ACV(k) = \left\|\left[\sum_{\tau = 0}^{k-1} \Gamma_r(\mu(\tau))\right]_+\right\|_2 
\]
\normalsize

Additional to resource sustainability behavior, we investigate the loss incurred to the population applying the resource allocation decisions $(\mu(\tau))_{\tau\in[k]_{0}}$ in form of the aggregated delay:

\small
\[
\AD(k) =  \sum_{\tau = 0}^k \sum_{i \in [N]}\Del_i(\tau),
\]
\normalsize
where $\Del_i(\tau)$ denotes the delay experienced by agent $i$ at time $\tau$, i.e. $\Del_i(\tau) = \sum_{\mathcal{P}_i \in \mathfrak{P}_i} \ell_{\mathcal{P}_i}^{(i)}(\mu(\tau)) \mu_{\mathcal{P}_i}^{(i)}(\tau)$.

It should be noted, that resource sustainability and loss minimization do not need to be coinciding objectives, but can display a trade-off behavior depending on model parameters, i.e. they appear as conflicting objectives. Therefore it can happen, that resource sustainability subsequently implies a disadvantaging of some agents.
\section{Resource-Centric Pricing for Congestion Game}

\paragraph*{Population Dynamic via Score and Hedge strategy}
Throughout this work, we consider the congestion game, which is played multiply with time horizon $n\in\nat$. We provide a summary of our model for the agents' decision-making process in Algorithm \ref{Alg:aoaishhjddhhddddeee2}.

\setlength{\textfloatsep}{3pt} 
\begin{algorithm}[htbp]
	\caption{Hedge algorithm with Prices}
	\begin{algorithmic}
		\REQUIRE $n\in\nat$, $\gamma>0$, $\Phi_{i}:\real^{\mathfrak{P}_{i}}\rightarrow\Delta_{i}$.
	\FOR{every agent $i\in [N]$ }
		\STATE{Initialize the score vector $Y_{0}^{(i)}\gets 0$}
	\ENDFOR
		\FOR{ time $k=1,2,\ldots,n$}
		\STATE Population apply the allocation strategy:
		\footnotesize $$X(k)=(m_{i}\mu^{(i)}(k))_{i\in [N]}$$\normalsize
		\FOR{every agent $i\in [N]$ }
		\STATE Receive the price vector $(\Lambda_{r}(k))_{r\in\mathcal{R}}$ broadcasted by the regulator.
		\FOR{ all bundle of resource $\mathcal{P}_{i}\in\mathfrak{P}_{i}$}
		\STATE Experience the disturbed cost:
		
		\footnotesize
		$$\hat{\ell}^{(i)}_{\mathcal{P}_{i}}(k)\gets \ell^{(i)}_{\mathcal{P}_{i}}(\mu(k))+\xi^{(i)}_{\mathcal{P}_{i}}(k+1)$$
		\normalsize
		
		\STATE Compute the price per amount of task:
		\footnotesize
		 \begin{equation*}
		\pi_{\mathcal{P}_{i}}^{(i)}(k)=\sum_{r\in\mathcal{P}_{i}}\Lambda_{r}(k)
		\end{equation*}
		\normalsize

			\STATE Update the score of bundle $\mathcal{P}_{i}$: 
			\footnotesize
			\begin{equation}
			\label{Eq:Update}
			 Y^{(i)}_{\mathcal{P}_{i}}(k+1)\gets Y^{(i)}_{\mathcal{P}_{i}}(k)-\gamma \left[ \hat{\ell}^{(i)}_{\mathcal{P}_{i}}(k)+\pi^{(i)}_{\mathcal{P}_{i}}(k)\right]
			 \end{equation}
			 \normalsize
			 
		\ENDFOR
		\STATE Generate the allocation strategy (see \eqref{Eq:aaooshshgdggdhhdgddhs2}):
		\footnotesize
		\begin{equation*}
		\mu^{(i)}(k+1)\gets \Phi^{(i)}(Y^{(i)}(k+1))
		\end{equation*}
		\normalsize
		\ENDFOR
		\ENDFOR 
	\end{algorithmic}
	\label{Alg:aoaishhjddhhddddeee2}
\end{algorithm}
\setlength{\textfloatsep}{0pt}

According to Algorithm \ref{Alg:aoaishhjddhhddddeee2}, every agent $i\in[N]$ accumulates at each round $k\in[n]$ the present and historical cost (discounted by a given parameter $\gamma$) of each resource bundle available to him, aiming to provide the scores of her bundle preferences. This model corresponds to non-myopic data-driven agents that utilize historical data to derive their strategy. This assumption of agents' behavior is quite plausible for recent applications that mostly utilize statistical and learning methods by accumulating past data (in this context: resource costs). 

 The corresponding actual cost of an available bundle consists of the actual noisy loss caused by negative externalities and the price set exogenously by a regulator (c.f. Algorithm 	\ref{Alg:aoaishhjddhhddddeee}). We assume that the nois$(\xi_{n})_{n\in\nat}$ is a $\real^{\sum_{i=1}^{N}\abs{\mathfrak{P}_{i}}}$-valued \textit{$\mathbb{F}$-martingale difference sequence} which is a quite general noise model.
One reason that we model the loss as noisy is that the environment or the imperfectness of agents' sensing devices can cause imperfectness of agents' feedback. Another reason is that we can handle the case where the agents' actions are discrete while their strategies are mixed states, and thus the resource congestion only represents an unbiased sample of the congestion specified by the mixed strategies (see \cite{Mertikopoulos2018}).


 The mapping $\Phi^{(i)}$ serves to model how the $i^{\text{th}}$ agent builds up his allocation strategy from the actual score of the bundles. In this work, we investigate the case where it takes the following specific form:
\begin{equation}
\label{Eq:aaooshshgdggdhhdgddhs2}
(\Phi^{(i)}(y^{(i)}))_{\mathcal{P}_{i}}=\tfrac{\exp(y^{(i)}_{\mathcal{P}_{i}})}{\sum_{\widetilde{\mathcal{P}}_{i}\in\mathfrak{P}_{i}}\exp(y^{(i)}_{\widetilde{\mathcal{P}}_{i}})}.
\end{equation}
Without altering the analysis given in this work, one can use more generally the concept of the mirror map (See also \cite{Mertikopoulos2018}) for specifying the choice map $\Phi^{(i)}(y^{(i)}):=\argmax\limits_{\mu^{(i)}\in\Delta_{i}}\left\{\left\langle \mu^{(i)},y^{(i)}\right\rangle-\psi_{i}(\mu^{(i)})\right\}$,
where $\psi^{(i)}:\Delta_{i}\rightarrow \real$ strongly convex w.r.t. $\norm{\cdot}_{1}$. For instance, one can use the usual Euclidean projection onto the simplex instead.

\begin{remark}
Without altering the analysis given in this work, one can use more generally the concept of mirror map (See also \cite{Mertikopoulos2018}) for specifying the choice map $\Phi_{i}$. Specifically $\Phi_{i}$ which takes the form 
\begin{equation*}
\Phi^{(i)}(y^{(i)}):=\argmax\limits_{\mu^{(i)}\in\Delta_{i}}\left\{\left\langle \mu^{(i)},y^{(i)}\right\rangle-\psi_{i}(\mu^{(i)})\right\},
\end{equation*}
for a function $\psi^{(i)}:\Delta_{i}\rightarrow \real$ strongly convex w.r.t. $\norm{\cdot}_{1}$. For instance the usual Euclidean projection onto the simplex can be used as the choice map.
\end{remark}
\paragraph*{Pricing Algorithm}
\setlength{\textfloatsep}{0pt} 
\begin{algorithm}[htbp]
	\caption{Resource-Centric Pricing}
	\begin{algorithmic}
		\REQUIRE $n\in\nat$, $\beta>0$, $\alpha\in (0,1]$  
		\STATE Initialize the price vector $\Lambda_{0}\gets 0$
		\FOR{ time $k=1,2,\ldots,n$}
		%
		%
		%
		\FOR {$r\in\mathcal{R}$}
		\STATE Check the actual load $\phi_{r,k}:=\phi_{r}(X(k))$ of ressource $r$ caused by Algorithm \ref{Alg:aoaishhjddhhddddeee2} 
		\STATE Update the price of resource $r$:
		\begin{equation}
		\label{Eq:aiaiahshshsss}
		\Lambda_{r}(k+1)\leftarrow\left[ (1-\alpha)\Lambda_{r}(k)+\beta\left(\phi_{r,k}-L_{r}\right)\right]_{+} 
		\end{equation}
		\ENDFOR 
		\ENDFOR
	\end{algorithmic}
	\label{Alg:aoaishhjddhhddddeee}
\end{algorithm}
\setlength{\textfloatsep}{0pt} 
To encourage sustainable use of the resources, we specify the price vector required by Algorithm \ref{Alg:aoaishhjddhhddddeee2} via the mechanism described in Algorithm \ref{Alg:aoaishhjddhhddddeee}. We underline the fundamental role of the price to reflect the scarcity of a resource by setting the price update \eqref{Eq:aiaiahshshsss} proportional to the present congestion state $\phi_{r,k}-L_{r}$ (with the parameter $\beta$ specifying the sensitivity of the prices to the congestion state). This aspect allows the regulators to indicate a possible resource overload implicitly.

Furthermore, we introduce "memory" into the price dynamics by involving the previous price update $\Lambda_{r}(k)$ into \eqref{Eq:aiaiahshshsss}. The reason is twofold; firstly, to ensure the alignment of the incentives with the non-myopic behavior of the agents, and secondly to track the congestion dynamic for analytical purposes. The latter reason becomes clear by iterating \eqref{Eq:aiaiahshshsss} (with $\alpha=0$), and recognizing, that the prices give an upper bound for the $\ACV$. However, a possible drawback of this procedure is that a sharp price increase might result in a domination of the agents' preferences (expressed by their losses): It follows from \eqref{Eq:Update} that unusually high prices caused the agents to decide for the resources, having the lowest prices and not for the ones giving them the lowest loss -- resulting in a degradation of the population's welfare. Thus, we introduce in Algorithm \ref{Alg:aoaishhjddhhddddeee} the parameter $\alpha$ whose role is to bypass the phenomenon above by offsetting the memory in the price dynamic. 
 	\begin{figure}[htbp]
 	\centering
 	\includegraphics[scale=0.7]{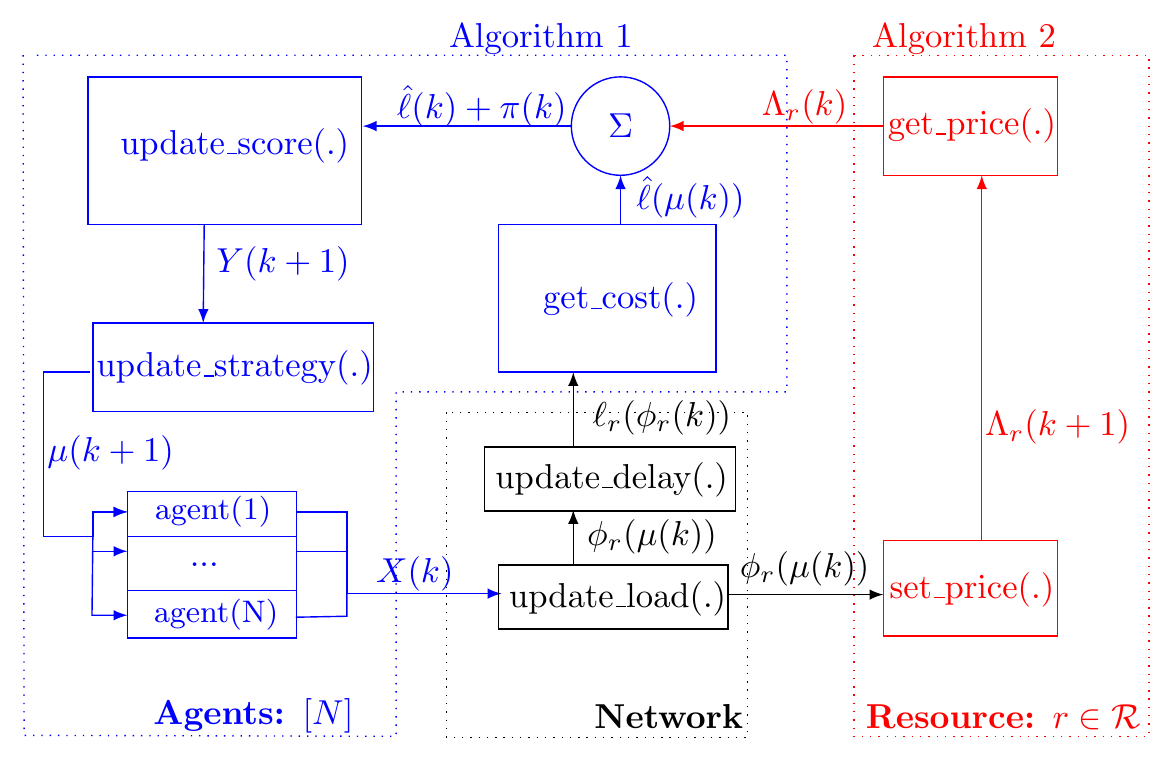}
 	\caption{Sketch of Algorithms 1 and 2. The blue color marks the agents' affairs, the red price the setters' (resources), and the black the networks'}
 	\label{fig:my_label}
 \end{figure}
 \paragraph*{Relation between Algorithms \ref{Alg:aoaishhjddhhddddeee2} and \ref{Alg:aoaishhjddhhddddeee}}
 In order to clarify the relationship between the price setters, i.e., resources, we sketch the connection between Algorithms \ref{Alg:aoaishhjddhhddddeee2} and \ref{Alg:aoaishhjddhhddddeee} in Figure   \ref{fig:my_label}.
  \begin{figure}[htbp]
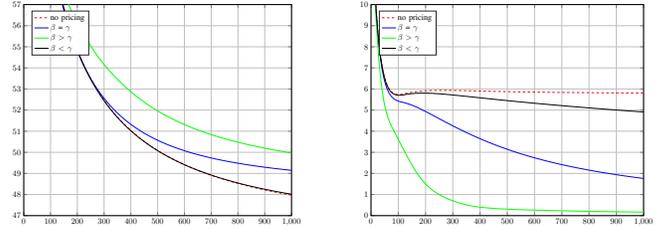
	
 	
 	%
 	\begin{minipage}{.48\linewidth}
 		\centering
 		\scalebox{.31}{\input{AD.tex}}
 		\centerline{\footnotesize(a) AD averaged over time}
 		
 	\end{minipage}
 	\hfill
 	\begin{minipage}{0.48\linewidth}
 		\centering
 		\scalebox{.31}{\input{AACV.tex}}
 		\centerline{\footnotesize(b) ACV averaged over time}
 	\end{minipage}
 	\caption{Performance for $L_{r}=14$}
 	\label{fig:Perf}
 \end{figure} It is apparent that the resource-prices are decided parallelly in-situ and do not require any centralized instance in contrast to most of the resource control mechanisms such as the 
 bidding-based and auction-based mechanism.
 This aspect is an advantage since centralized solutions are known to be sensitive to malicious attacks and require rather sophisticated computations, e.g., solving an optimization problem. Also, we want to stress that the price of a resource $r$ is based purely on the congestion state $\phi_{r,k}-L_{r}$ and not on the (preferences of the) agents utilizing the resource $r\in\mathcal{R}$.Since the agents do not have to reveal their strategy and preferences (e.g., resource bundles), our method respects the sovereignty and the privacy of the individuals. Moreover, since our method does not cling to a specific agent's feedback, agents can be added or removed, making this approach particularly flexible. 
 
By not knowing the preferences of the agents and due to the absence of a centralized instance, we may sacrifice some desired properties of the mechanism (e.g., driving the population toward a socially optimum state and budget balance). However, in order to approach the fulfillment of the first property, tuning the parameters of the mechanism and accepting looser resource constraints so that the prices do not dominate the loss of the agents, results in resource sustainability with a lower cost of welfare degradation (for details see Section \ref{Sec:aiaisshshssggsss}).    

\section{Performance Analysis}
\label{Sec:PerfAn}

Throughout, $C_{1},C_{2},C_{3},m_{*}$ denote non-negative constants fulfilling for all $\mu\in\Delta$ and $\lambda\in\real^{\mathcal{R}}_{\geq 0}$:
\begin{equation*}
\begin{split}
&\sum_{i=1}^{N}m_{i}\norm{\M^{(i),\T}\lambda}^{2}_{\infty}\leq C_{1}^{2}\norm{\lambda}^{2}_{2},~\sum_{i=1}^{N}m_{i}\norm{\ell^{(i)}(\mu)}_{\infty}^{2}\leq C_{2}^{2}\\
&\norm{\phi(\mu)-L}_{2}\leq C_{3},\quad m_{i}\leq m_{*}, ~\forall i\in[N] .
\end{split}  
\end{equation*}

Our main result is the following:
\begin{theorem}
\label{Thm:oajjshhdggfggfhhdggdhhd3}
Let $\gamma>0$ be given, $\beta=\gamma$, and $\alpha=\delta \gamma^{2}$ with $\delta>0$ satisfying  \begin{equation}
\label{Eq:aiaiahhfggfdhdgdhdhdgsgss}
(C_{1}^{2}+\gamma^{2}\delta^{2})-\frac{\delta}{2}\leq 0.
\end{equation} 
It holds:
\begin{equation}
\label{Eq:aiaiahssgsghshhsgsgs}
\begin{split}
\Erw\left[\tfrac{\norm{\Lambda(n)-\lambda_{*}}_{2}^{2}}{2}\right]\leq&\tfrac{\Delta\psi^{2}}{2}+(1+\alpha n)\tfrac{\norm{\lambda_{*}}_{2}^{2}}{2}+
\tfrac{\tilde{C}_{1}^{2}}{2}\gamma^{2}n\\
&+2\gamma^{2}m_{*}N\sum_{k=1}^{n}\Erw[\norm{\xi_{k}}^{2}_{\infty}]
\end{split}
\end{equation}
where $\widetilde{C}_{1}^{2}:=2\left(C_{2}^{2}+2C_{3}^{2}\right)$ and $\Delta\psi^2=2m_{*}\sum_{i=1}^{N}\ln(\abs{\mathfrak{P}_{i}})$
\end{theorem}
\begin{remark}
A necessary condition for gamma such that there exists a $\delta>0$ satisfying \eqref{Eq:aiaiahhfggfdhdgdhdhdgsgss} is:
\begin{equation}
\label{Eq:aoaosjsjskdjdjdjdhff}
\gamma\leq\frac{1}{4C_{1}}.
\end{equation} 
If this is fulfilled, then \eqref{Eq:aiaiahhfggfdhdgdhdhdgsgss} is equivalent to:
\begin{equation}
\label{Eq:aaoaosjjshdhdhhdhddd}
 \frac{1-\sqrt{1-16\gamma^{2}C_{1}^{2}}}{4\gamma^{2}}\leq\delta\leq \frac{1+\sqrt{1-16\gamma^{2}C_{1}^{2}}}{4\gamma^{2}}.
\end{equation}
We also observe that for small enough $\gamma$, we can choose $\delta\approx 2 C_1$, which does not depend on the horizon length. Attentive reader may recognize by inspecting the proof of above Theorem that in order that above result holds, it is not necessary, that $\alpha$ is of the form $\alpha=\delta\gamma^{2}$, and thus that the regulator knows precisely about the agents' step size. The only requirement is that $\alpha$ has to decay slower than $\gamma^{2}$ with the time horizon $T$. However, one obtains the best rate for the performance guarantee in case that $\alpha$ is of order $\gamma^{2}$ (w.r.t. $T$). 
\end{remark}
The proof of Theorem \ref{Thm:oajjshhdggfggfhhdggdhhd3} is given in the Appendix.

An immediate consequence of Theorem \ref{Thm:oajjshhdggfggfhhdggdhhd3} is the following (for proof see the full version \cite{Tamp2020Cong}) guarantee for the accumulation of the capacity violation:
\begin{corollary}
\label{Corr:aiaahssgsgshhsssddd}
Suppose that the conditions of Theorem \ref{Thm:oajjshhdggfggfhhdggdhhd3} are fulfilled and that the noise is persistent in the sense that there exists $\sigma^{2}>0$ s.t. $\Erw[\norm{\xi_{k}}_{\infty}^{2}]\leq \tfrac{\sigma^{2}}{4m_{*}N}$ for all $k\in \nat$.
It holds:
\begin{equation}
\label{Eq:aiaiahssgsghshhsgsgs2}
\begin{split}
\Erw\left[\norm{\Lambda(n)}_{2}\right]\leq& \Delta\psi+(1+\sqrt{(1+\delta\gamma^{2} n)})\norm{\lambda_{*}}_{2}\\
&+
(\tilde{C}_{1}+\sigma)\gamma\sqrt{n},
\end{split}
\end{equation}
where $\Delta\psi$ and $\tilde{C}_{1}$ is given as in Theorem \ref{Thm:oajjshhdggfggfhhdggdhhd3}. 
Now, suppose that $\gamma:=c/\sqrt{n}$:
for a constant $c>0$ and $\delta\in (0,1/\gamma^{2})$ s.t. \eqref{Eq:aiaiahhfggfdhdgdhdhdgsgss} is fulfilled. It holds:
\begin{equation}
\label{Eq:aoaoshshggsgsgsffdgdd}
\begin{split}
\Erw\left[ \ACV(n)\right]\leq(\delta c+\tfrac{1}{c}) A\sqrt{n}, 
\end{split}
\end{equation}
where $A:=\Delta\psi+(1+\sqrt{(1+\delta c^{2})})\norm{\lambda_{*}}_{2}+
(\tilde{C}_{1}+\sigma)c$.
\end{corollary}
\section{Simulation}
\label{Sec:aiaisshshssggsss}
\textbf{Game Setting:} We consider the network routing problem given in Example \ref{Ex:Network Routing Game} which we specify as follows: $\mathcal{V}$ consists of 15 nodes and $\mathcal{E}$ is built from a randomly generated adjacency matrix (without self-loop) with independent entries, where each non-diagonal is $1$ with probability $0.5$. Furthermore, we consider $N = 10$ agents, each has the starting point and the destination randomly uniformly chosen from $\mathcal{V}$. Given the latter, each agent $i$ has randomly created bundles of maximal size $|\mathcal{P}_i| \leq 10$. We set the total resource load $m_i = 20$, $\forall i \in [N]$, and the admissible flow per resource $L_r = 14$, $\forall r \in \mathcal{R}$. For the cost per resource $\ell_{r,k}$, we consider a quadratic polynomial of the form $\ell_{r,k}(\phi_r(k)) = a_2^{(r)}\phi_r(k)^2+a_1^{(r)}\phi_r(k)+a^{(r)}_0$, where the coefficients $(a_2^{(r)},a_1^{(r)},a_0^{(r)})$ for each resource $r \in \mathcal{R}$ are independently randomly uniformly chosen from $[0, 0.05]$.

\textbf{Parameter Setting:} We set the parameters required by Algorithms 1 and 2 as follows: We consider the time horizon $n=10^{3}$, the agents' learning rate $\gamma=0.1\sqrt{n}=0.0032$, and the response parameter $\alpha=10^{-5}$. We are not only interested in the case $\beta=\gamma$ analyzed in Section \ref{Sec:PerfAn}, but also in the case where the regulator is uncertain about the agents' learning rate, and therefore $\beta$ differs significantly from $\gamma$ by the factor 10: $\beta=10\gamma$ ($\beta>\gamma$) and $\beta=10^{-1}\gamma$ ($\beta<\gamma$). For the noise modeling w.r.t. the disturbed cost we consider uniformly distributed random i.i.d. samples between [-0.01,0.01]. 
%
%

\textbf{Performance Evaluation:} Fig. \ref{fig:Perf} shows that our pricing mechanism reduces the aggregated capacity violation even if $\beta\neq \gamma$ since the $\ACV$ for each of the parameter choices is significantly lower than $\ACV$ of purely anarchistic case (red,dashed). However, we observe that a higher $\beta$ may accelerate this process. Additionally, we see that our pricing method does not yield significant discrimination of the agents, when compared to the improvement of the capacity violation, as the differences between the aggregated delays for the different cases are marginal at worst (see Fig. \ref{fig:Perf} (b)). Still, we note a trade-off behavior in the choice of $\beta$: In case that $\beta$ is high ($\beta >\gamma$), the capacity violation is the lowest, but the experienced delay the highest.\begin{figure}[htbp]
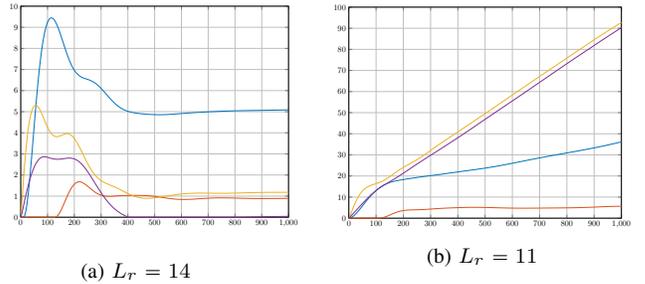

	\begin{minipage}{.48\linewidth}
		\centering
		\scalebox{.31}{\input{price.tex}}
		\centerline{\footnotesize(a) $L_{r}=14$}
	\end{minipage}
	\hfill
	\begin{minipage}{0.48\linewidth}
		\centering
		\scalebox{.31}{\input{price3.tex}}
		\centerline{\footnotesize(b) $L_{r}=11$}
	\end{minipage}
	\caption{Pricing over time}
	\label{fig:Perf3}
\end{figure} This occurrence reflects the increasing dominance of the price regulation over the agents' personal interest to decrease the incurred delay.
 Another observation which we make is that if $\beta=\gamma$, some prices might at worst be constant for large times as predicted in Corollary \ref{Corr:aiaahssgsgshhsssddd}, indicating that even if the population fulfills resource constraints, a control mechanism is necessary to maintain this desired status quo.  
%
%
%
%



\textbf{Overly Strict Capacity Constraints:} We also investigate the performance of our method with stricter capacity constraints, i.e. $L_{r}=11$. We see that our method still yields an improvement of the capacity violation compared to the no pricing case (see Fig. \ref{fig:Perf2}). However, this comes with a significant reduction of agents' welfare in the form of a higher $\AD$ (see Fig. \ref{fig:Perf2} (a)). One may justify this as follows: Taking a look at the pricing evolution (Fig. \ref{fig:Perf3} (b)) of exemplary resources, we observe a linear increase in prices dominating the personal preferences ($\hat{\ell}^{(i)}_{\mathcal{P}_{i}}$ in \eqref{Eq:Update}) of the agents in large times. Consequently, each of the affected agents decides for routes that have the lower prices rather than those that incur the lowest delay. 

The enormous increase of prices shown in Fig. \ref{fig:Perf3} (b) gives a hint that the minimizer of the Rosenthal potential corresponding to the network routing game over $\mathcal{Q}$ does not exist (c.f. the Proof of Theorem \eqref{Thm:oajjshhdggfggfhhdggdhhd3}) due to overly strict resource constraints. However, one may able to show the sub-linearity of $\ACV$ to be of order $\mathcal{O}(n^{1/4})$. Moreover, the increase in prices is in contrast to the case where the capacity constraints are rather loose (Fig. \ref{fig:Perf3} (a)). The latter observations give the following heuristic: In case that one observes a linear increase of some prices, one may set a looser constraint so that the reduction of the capacity violations does not come with a significant reduction of the populations' welfare.

\section{Summary, Discussion, and Future Work}
Assuming that the agents are choosing their action based on the average historical cost of the resource bundles 
and the logit choice rule, we introduced a resource-centric pricing mechanism which allows a non-asymptotic guarantee of the sub-linear growth 
of the expected aggregated violation of the resource constraints of order $\mathcal{O}(\sqrt{n})$.
\begin{figure}[htbp]
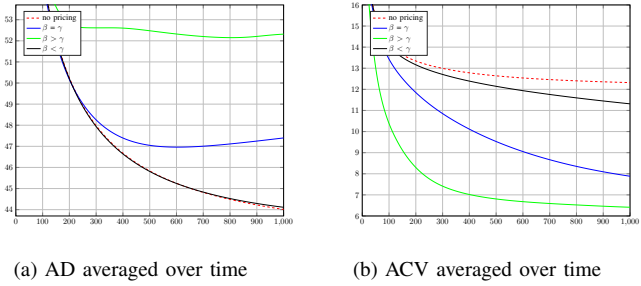

	
	%
	\begin{minipage}{.48\linewidth}
		\centering
		\scalebox{.31}{\input{AD2.tex}}
		\centerline{\footnotesize(a) AD averaged over time}
	\end{minipage}
	\hfill
	\begin{minipage}{0.48\linewidth}
		\centering
		\scalebox{.31}{\input{AACV2.tex}}
		\centerline{\footnotesize(b) ACV averaged over time}
	\end{minipage}
	\caption{Performance for $L_{r}=11$}
	\label{fig:Perf2}
\end{figure}
 In case that the resource constraints are not overly strict, we observe numerically that the resource sustainability delivered by our method, does not come with significant discrimination of the agents. For the general case, trade-off effect between resource sustainability and population's welfare might occur. In the future, we plan to explain these aspects formally. 
\bibliographystyle{IEEEtran}
\bibliography{BibWardGlob}
\appendix
\subsection{Additional Notations}
\begin{itemize}
\item $\widetilde{\M}^{(i)}=m_{i}\M^{(i)}$ 
\item One can express $\phi(\mu)$ more compactly by $\phi(\mu)=\widetilde{\M}\mu$, where $\widetilde{\M}=[\widetilde{\M}^{(1)}|\cdots|\widetilde{\M}^{(N)}]$
\end{itemize}
\subsection{Auxiliary Statements}
\begin{lemma} 
	\label{Lem:UpperBoundbyPrice}
	Suppose that $\Lambda_{0}=0$.
	For all $r\in [R]$ and $k\in\nat$:
	\begin{equation}
	\label{Eq:aaisjjshshsgdgdd}
	\ACV(k) \leq \frac{\norm{\Lambda(k)}_{2}+\alpha\sum_{\tau=1}^{k-1}\norm{\Lambda(\tau)}_{2}}{\beta}
	\end{equation}
\end{lemma}
\subsubsection{Proof of Lemma \ref{Lem:UpperBoundbyPrice}}
\label{Subsec:Lema}
The definition of our price policy gives: 
\begin{align*}
\Lambda_{r}(\tau+1)\geq\Lambda_{r}(\tau)+\beta\Gamma_{r}(\tau)-\alpha\Lambda_{r}(\tau).
\end{align*}
So by summing this inequality and subsequent telescoping, we have since $\Lambda_{0}=0$: $$\sum_{\tau=0}^{k-1}\Gamma_{r}(\tau) \leq \frac{\Lambda_{r}(k)+\sum_{\tau=1}^{k-1}\alpha\Lambda_{r}(\tau)}{\beta}.$$ 
We observe that the inequality also holds when applying $[\cdot]_{+}$ to the L.H.S. due to the R.H.S. of the inequality being non-negative and since $[\cdot]_{+}$ is monotonically increasing. Moreover, since the resulting inequality holds for all $r\in [R]$ and both sides of the equation are positive it follows using monotonicity:
\begin{equation*}
\left\|\left[\sum_{\tau = 0}^{k-1} \Gamma(\tau)\right]_+\right\|_2 \leq \left\| \frac{\Lambda(k)+\sum_{\tau=1}^{k-1}\alpha\Lambda(\tau)}{\beta}\right\|_2
\end{equation*}
We notice, that the L.H.S. corresponds to our definition of the aggregated capacity violation at stage $k-1$. Therefore applying the triangle inequality yields the upper-bound \eqref{Eq:aaisjjshshsgdgdd} and (subsequently an approximation) for the aggregated capacity violation purely based on the price.
\subsection{Monotonicity of the KKT-operators}
\label{Subsec:Monotone}
An operator $F:\real^{D}\rightarrow\real^{D}$ is said to be monotone on $\mathcal{Z}\subseteq\real^{D}$ if $\inn{x_{1}-x_{2}}{F(x_{1})-F(x_{2})}\leq 0$, for all $x_{1},x_{2}\in \mathcal{Z}$. If in the latter strict inequality hold for $x_{1}\neq x_{2}$, then $F$ is said to be strictly monotone. $F$ is said to be $c$-strongly monotone on $\mathcal{Z}$ if $\inn{x_{1}-x_{2}}{F(x_{1})-F(x_{2})}\leq -c \norm{x_{1}-x_{2}}^{2}$, for all $x_{1},x_{2}\in \mathcal{Z}$. 
\begin{proposition}
\label{Prop:aiajshsgdggdgdgdhsgss}
Let be $\X\subset \real^{D}$.
Consider the operator $\tilde{v}:\mathcal{X}\times \real^{M}_{\geq 0}\rightarrow \real^{D}\times\real^{M}$ given by:
\begin{equation}
\label{Eq:iaaiisjsjdhhdhdhdhdjjsss}
(x,\lambda)\mapsto \left[v(x)+A^{\T}\lambda, b-Ax\right]^{\T},
\end{equation}
where $v:\X\rightarrow\real^{D}$, $A\in\real^{M\times D}$, and $b\in\real^{M}$. It holds:
\begin{equation}
\label{Eq:aiaiahshjshdhjddd}
\inn{\alpha_{1}-\alpha_{2}}{\tilde{v}(\alpha_{1})-\tilde{v}(\alpha_{2})}=\inn{x_1-x_{2}}{v(x_{1})-v(x_{2})},
\end{equation}
for all $\alpha_{i}:=(x_{i},\lambda_{i})\in\mathcal{X}\times \real^{M}_{\geq 0}$, $i=1,2$. 
\end{proposition}
\begin{proof}
Straightforward computations yields:
\begin{align*}
&\inn{\alpha_{1}-\alpha_{2}}{\tilde{v}(\alpha_{1})-\tilde{v}(\alpha_{2})}=\inn{x_{1}-x_{2}}{v(x_{1})-v(x_{2})}\\
&+\inn{x_{1}-x_{2}}{A^{T}\lambda_{1}-A^{T}\lambda_{2}}-\inn{\lambda_{1}-\lambda_{2}}{Ax_{1}-A x_{2}}.
\end{align*}
Moreover, we have:
\begin{align*}
\inn{\lambda_{1}-\lambda_{2}}{Ax_{1}-Ax_{2}}=\inn{A^{T}\lambda_{1}-A^{T}\lambda_{2}}{x_{1}-x_{2}},
\end{align*}
Combining both computations, we obtain \eqref{Eq:aiaiahshjshdhjddd}.
\end{proof}
\subsection{Proof of the main result}
\begin{proof}[Proof of Theorem \ref{Thm:oajjshhdggfggfhhdggdhhd3}]
The logit choice $\Phi^{(i)}$ given in \eqref{Eq:aaooshshgdggdhhdgddhs2} is a mirror map (Definition 3.1 in \cite{Mertikopoulos2018}) induced by the negative Gibbs entropy $\psi_{i}(\mu^{(i)})=\sum_{\mathcal{P}_{i}\in\mathfrak{P}_{i}}\mu^{(i)}_{\mathcal{P}_{i}}\ln(\mu^{(i)}_{\mathcal{P}_{i}})$
	as regularizer on the simplex which is a compact convex subset. Let be $F^{\mathbf{m}}(\mu,Y(k)):=\sum_{i=1}^{N}m_{i}F_{i}(\mu^{(i)},Y^{(i)}(k))$ where $F_{i}$ is the Fenchel coupling (Definition 4.2 in \cite{Mertikopoulos2018}) induced by the negative Gibbs entropy as $1$-strongly (w.r.t. $\norm{\cdot}_{\infty}$) convex regularizer on the simplex $\Delta_{i}$. 

By means of $F^{\mathbf{m}}$, we can estimate the evolution of Algorithm 	\ref{Alg:aoaishhjddhhddddeee2} with the dynamic pricing mechanism given in Algorithm 	\ref{Alg:aoaishhjddhhddddeee} by means of Lyapunov's type argumentation. Toward this end, we use the usual bound for the one step difference of the Fenchel coupling (see e.g. Proposition 4.3 (c) in \cite{Mertikopoulos2018}), insert the given iterate at time $k+1$ in the resulted inequality, and apply triangle inequality, to obtain:
\begin{equation}
\label{Eq:aiaisjsshshshhsss}
\begin{split}
&F^{\mathbf{m}}(\mu,Y(k+1))-F^{\mathbf{m}}(\mu,Y(k))\\
&\leq-\gamma\underbrace{\sum_{i=1}^{N}m_{i}\inn{\mu^{(i)}(k)-\mu^{(i)}}{\hat{\ell}^{(i)}(k)+\pi^{(i)}(k)}}_{=:\text{(a)}}\\
&+\frac{\gamma^{2}}{2}\underbrace{\sum_{i=1}^{N}m_{i}\norm{\hat{\ell}^{(i)}(k)+\pi^{(i)}(k)}_{\infty}^{2}}_{=:\text{(b)}}.
\end{split}
\end{equation}
By the triangle inequality and the definition of the constants given in Section \ref{Sec:PerfAn}, we can estimate the summand (b) as follows:
\begin{equation}
\label{Eq:aiaiashhshsjshshsssss}
\text{(b)}/2\leq C_{1}^{2}\norm{\Lambda(k)}_{2}^{2}+2(C_{2}^{2}+\sum_{i=1}^{N}m_{i}\norm{\xi^{(i)}_{k+1}}_{\infty}^{2})
\end{equation}
Now to estimate the summand (a),
notice that we can write:
\begin{equation}
\label{Eq:aiaisjjshshshsss}
\sum_{i=1}^{N}m_{i}\inn{\mu^{(i)}(k)-\mu^{(i)}}{\pi^{(i)}(k)}=\inn{\mu(k)-\mu}{\widetilde{\M}^{T}\Lambda(k)}.
\end{equation}

Combining all the previous observations,  we have by summing the resulting inequality over all $k=0,\ldots,n-1$, and by subsequent telescoping,
we obtain an upper bound for the cumulative difference $\mathcal{V}_{n}^{(1)}(\mu):=F^{\mathbf{m}}(\mu,Y(n))-F^{\mathbf{m}}(\mu,Y(0))$:
\begin{equation}
\label{Eq:aaoosjjssjhdhdhjjssh}
\begin{split}
\mathcal{V}_{n}^{(1)}(\mu)
\leq&-\gamma\sum_{k=0}^{n-1}\underbrace{\sum_{i=1}^{N}\underbrace{m_{i}\inn{\mu^{(i)}(k)-\mu^{(i)}}{\ell^{(i)}(\mu(k))}}_{=\inn{\mu^{(i)}(k)-\mu^{(i)}}{\nabla_{\mu^{(i)}(k)}V(\mu(k))}}}_{\inn{\mu(k)-\mu}{v(\mu(k))}}\\
&-\gamma\sum_{k=0}^{n-1}\inn{\mu(k)-\mu}{\widetilde{\M}^{T}\Lambda(k)}
\\
&+\gamma^{2}C_{1}^{2}\sum_{k=0}^{n-1}\norm{\Lambda(k)}^{2}_{2}+\gamma S_{n}+2\gamma^{2}R_{n}+2 C_{2}^{2}\gamma^{2}n
\end{split}
\end{equation}
where:
\begin{equation*}
S_{n}:=-\sum_{k=0}^{n-1}\inn{X(k)-x_{*}}{\xi(k+1)},~R_{n}:=m_{*}N\sum_{k=1}^{n}\norm{\xi(k)}^{2}_{\infty},
\end{equation*}
and where $v(\mu):=\nabla V(\mu)$, where $V$ denotes the Rosenthal potential:
\begin{equation}
\label{Eq:aaiiahshhsjhggfhhfggff}
V:\Delta\rightarrow\real,\quad \mu\mapsto \sum_{r\in\mathcal{R}}\int_{0}^{\phi_{r}(\mu)} \ell_{r}(u)\d u,
\end{equation}

We now estimate the evolution of the price vector by providing a bound for $\mathcal{V}_{n}^{(2)}(\lambda):=(\norm{\Lambda(n)-\lambda}_{2}^{2}-\norm{\Lambda(0)-\lambda}_{2}^{2})/2$, where $\lambda\geq 0$. By similar computations as before, 
and by the elementary bound $2\inn{\lambda-\Lambda(k)}{\Lambda(k)}\leq\norm{\lambda}^{2}_{2}-\norm{\Lambda(k)}^{2}_{2}$,
we obtain: 
\begin{equation}
\begin{split}
&\mathcal{V}_{n}^{(2)}(\lambda)\leq \beta\sum_{k=0}^{n-1}\inn{\Lambda(k)-\lambda}{\phi(\mu(k))-L}\\&+\tfrac{\alpha}{2}\sum_{k=0}^{n-1}(\norm{\lambda}_{2}^{2}-\norm{\Lambda(k)}_{2}^{2})
+\sum_{k=0}^{n-1}(\beta^{2}C_{3}^{2}+\alpha^{2}\norm{\Lambda(k)}_{2}^{2}).
\end{split}
\label{Eq:aaiaiajssjshhdddd}
\end{equation}

Combining the bounds \eqref{Eq:aaoosjjssjhdhdhjjssh} and \eqref{Eq:aaiaiajssjshhdddd},
, it holds:
\begin{equation}
\label{Eq:aiaiaiisjsjsss}
\begin{split}
&\mathcal{V}_{n}^{(1)}(\mu)+\mathcal{V}_{n}^{(2)}(\lambda)\\
&\leq-\gamma\sum_{k=0}^{n-1}\inn{z(k)-z}{\tilde{v}(z(k))}\\
& +(\beta-\gamma)\sum_{k=0}^{n-1}\inn{\Lambda(k)-\lambda}{\widetilde{\M}\mu(k)-L}\nonumber\\
&+(\gamma^{2}C_{1}^{2}-\frac{\alpha}{2}+\alpha^{2})\sum_{k=0}^{n-1}\norm{\Lambda(k)}_{2}^{2}\\
&+\left(2C_{2}^{2}\gamma^{2}+C_{3}^{2}\beta^{2}+\frac{\alpha \norm{\lambda}_{2}^{2}}{2}\right)n+\gamma S_{n}+2\gamma^{2}R_{n},
\end{split}
\end{equation}
where:
\begin{equation*}
\begin{split}
&z(k):=(\mu(k),\Lambda(k)),\quad z=(\mu,\lambda),\\&~\tilde{v}(z(k))=[\nabla V(\mu(k))+\widetilde{\M}^{T}\Lambda(k),L-\widetilde{\M}\mu(k)].
\end{split}
\end{equation*}
setting $\beta=\gamma$ and
$\alpha=\delta\gamma^{2}$ with $\delta\in (0,1/\gamma^{2})$ fulfilling \eqref{Eq:aiaiahhfggfdhdgdhdhdgsgss},
we have:
\begin{equation}
\label{Eq:aaisjsshdhddgggddd}
\begin{split}
&\mathcal{V}_{n}^{(1)}(\mu)+\mathcal{V}_{n}^{(2)}(\lambda)\\
&\leq-\gamma\sum_{k=0}^{n-1}\underbrace{\inn{z(k)-z}{\tilde{v}(z(k))}}_{=:\Upsilon_{k}(z,z(k))}\\
&+\left((2C_{2}^{2}+C_{3}^{2})\gamma^{2}+\frac{\alpha \norm{\lambda}_{2}^{2}}{2}\right)n+\gamma S_{n}+2\gamma^{2}R_{n}.
\end{split}
\end{equation}
Notice that $v$ is monotone (for the definition of monotone operator see \ref{Subsec:Monotone}) since $V$ is convex. Thus by Proposition \ref{Prop:aiajshsgdggdgdgdhsgss}, we have that $\tilde{v}$ is also monotone implying:
\begin{equation*}
    \Upsilon_{k}(z,z(k))\geq \inn{z(k)-z}{\tilde{v}(z)}.
\end{equation*}
Moreover, by the slater's condition and KKT argumentations, we can find a Lagrangian dual optimizer $\lambda_{*}\in\real^{\mathcal{R}}_{\geq 0}$ corresponding to the minimizer $\mu_{*}$ of $V$ over $\mathcal{Q}:=\lrbrace{\mu\in\Delta:~\Gamma(\mu)\leq 0}$. It follows that $(\mu_{*},\lambda_{*})\in \text{SOL}(\X\times\real^{\mathcal{R}},\tilde{v})$, and consequently:
\begin{equation}
\label{Eq:aaooashshsgggssgss}
    \Upsilon_{k}(z_{*},z(k))\geq \inn{z(k)-z_{*}}{\tilde{v}(z_{*})}\geq 0,~z_{*}=(\mu_{*},\lambda_{*}).
\end{equation}
Setting this observation into \eqref{Eq:aaisjsshdhddgggddd}, we obtain:
\begin{equation}
\label{Eq:aaisjsshdhddgggddd2}
\begin{split}
&\mathcal{V}_{n}^{(1)}(\mu_{*})+\mathcal{V}_{n}^{(2)}(\lambda_{*})\\
&\leq\left((2C_{2}^{2}+C_{3}^{2})\gamma^{2}+\frac{\alpha \norm{\lambda_{*}}_{2}^{2}}{2}\right)n+\gamma S_{n}+2\gamma^{2}R_{n}.
\end{split}
\end{equation}

Now, since $Y_{0}=0$, we have:
\begin{equation*} 
\mathcal{V}^{(1)}_{n}(\mu_{*})\geq -\sum_{i=1}^{N} m_{i}\left( \max_{\Delta_{i}}\psi_{i}-\min_{\Delta_{i}}\psi_{i}\right) \geq -m_{*}\sum_{i=1}^{N}\ln(\abs{\mathfrak{P}_{i}}),
\end{equation*}
and thus $\mathcal{V}^{(1)}_{n}(\mu_{*})\geq -\Delta\psi^{2}/2$.
Combining this observation with \eqref{Eq:aaisjsshdhddgggddd2} and using $\Lambda_0 = 0$, we obtain that:
\begin{equation}
\label{Eq:aiaiashshshgdffsss}
\begin{split}
& \tfrac{\norm{\Lambda(n)-\lambda_{*}}_{2}^{2}}{2}
\\
&\leq  \tfrac{\Delta\psi^{2}}{2}+\underbrace{\tfrac{\norm{\Lambda(0)-\lambda_{*}}_{2}^{2}}{2}}_{\tfrac{\norm{\lambda_{*}}_{2}^{2}}{2}}+
\left((2C_{2}^{2}+C_{3}^{2})\gamma^{2}+\frac{\alpha \norm{\lambda_{*}}_{2}^{2}}{2}\right)n\\
&+\gamma S_{n}+2\gamma^{2}R_{n}\\
&=\tfrac{\Delta\psi^{2}}{2}+(1+\alpha n)\tfrac{\norm{\lambda_{*}}_{2}^{2}}{2}+
(2C_{2}^{2}+C_{3}^{2})\gamma^{2}n\\
&+\gamma S_{n}+2\gamma^{2}R_{n}
\end{split}
\end{equation}
Since $S_{n}$ is a martingale with $\Erw[S_{1}]=0$, we have by taking the expectation (and noticing $\Erw[S_{n}]=0$), the desired result.
\end{proof}
\subsection{Proof of consequences of the main result}
\begin{proof}[Proof of Corollary \ref{Corr:aiaahssgsgshhsssddd}]
Jensen's and triangle inequality asserts that:
\begin{equation*}
\sqrt{\Erw\left[\norm{\Lambda(n)-\lambda_{*}}_{2}^{2}\right]}\geq\Erw\left[\norm{\Lambda(n)-\lambda_{*}}_{2}\right]\geq \Erw\left[\norm{\Lambda(n)}_{2}\right]-\norm{\lambda_{*}}_{2}
\end{equation*}
Applying this to \eqref{Eq:aiaiahssgsghshhsgsgs} and by the persistence of the noise, we obtain \eqref{Eq:aiaiahssgsghshhsgsgs2}.

For any $k\in [n]$, we have by Corollary \ref{Corr:aiaahssgsgshhsssddd}:
\begin{equation*}
\begin{split}
\Erw\left[\norm{\Lambda(k)}_{2}\right]
\leq& \Delta\psi+(1+\sqrt{(1+\delta\gamma^{2} n)})\norm{\lambda_{*}}_{2}\\
&+
(\tilde{C}_{1}+\sigma)\gamma\sqrt{n}.
\end{split}
\end{equation*}
Now, setting our choices of parameters into \eqref{Eq:aiaiahssgsghshhsgsgs2}, it yields:
\begin{equation*}
\begin{split}
\Erw\left[\norm{\Lambda(k)}_{2}\right]
\leq& \Delta\psi+(1+\sqrt{(1+\delta c^{2})})\norm{\lambda_{*}}_{2}\\
&+
(\tilde{C}_{1}+\sigma)c=A.
\end{split}
\end{equation*}
Consequently:
\begin{equation}
\label{Eq:aiaishshssgdgdgddd}
\begin{split}
\tfrac{\alpha}{\beta}\Erw\left[\sum_{k=0}^{n-1}\norm{\Lambda(k)}_{2}\right]&=\tfrac{\delta c}{\sqrt{n}}\sum_{k=1}^{n-1}\Erw\left[\norm{\Lambda(k)}_{2}\right]\\
&\leq \tfrac{\delta c (n-1)}{\sqrt{n}} A\leq \delta c A\sqrt{n}.
\end{split}
\end{equation}
Moreover, we have $\Erw\left[\norm{\Lambda(n)}_{2}\right]/\beta\leq A\sqrt{n}/c$.
Setting this observation and \eqref{Eq:aiaishshssgdgdgddd} into \eqref{Eq:aaisjjshshsgdgdd}, we have the remaining statement.
\end{proof}

\end{document}